\documentclass[11pt, english]{article}

\usepackage[T1]{fontenc}
\usepackage{mathpazo}

\usepackage[margin=1in,bottom=1in,top=1in]{geometry}

\usepackage{amsmath,amsthm,bm}
\usepackage{enumitem}

\usepackage{setspace}
\usepackage{comment}

\usepackage{natbib}

\usepackage{hyperref}
\hypersetup{
	colorlinks = true,
	linkcolor = blue,
	citecolor = blue,
}

\newtheorem{assumption}{Assumption}
\newtheorem{lemma}{Lemma}

\newtheorem{theorem}{Theorem}
\newtheorem{corollary}{Corollary}

\newtheorem{assumptionalt}{Assumption}[assumption]

\newenvironment{assumptionp}[1]{
  \renewcommand\theassumptionalt{#1}
  \assumptionalt
}{\endassumptionalt}

\newcommand{\norm}[1]{\left\lVert#1\right\rVert}

\DeclareMathOperator*{\argmin}{arg\,min} 

\begin{document}

\title{Uniform Convergence Results for the Local Linear Regression Estimation of the Conditional Distribution\thanks{This article is a revised version of the second chapter of the author's doctoral dissertation at UC San Diego. The author is indebted to the Editor, an anonymous Associate Editor, and an anonymous referee for very constructive and detailed suggestions that greatly improved the paper.}}

\author{Haitian Xie\thanks{%
Email: xht@gsm.pku.edu.cn. Address: 5 Yiheyuan Road, Haidian District, Beijing, China, 100871.} \\
Guanghua School of Management\\
Peking University}

\date{\textbf{\today}}

\maketitle
\thispagestyle{empty}
\vspace{-2em}

\begin{abstract} \doublespacing
	This paper examines the local linear regression (LLR) estimate of the conditional distribution function $F(y|x)$. We derive three uniform convergence results: the uniform bias expansion, the uniform convergence rate, and the uniform asymptotic linear representation. The uniformity in the above results is with respect to both $x$ and $y$ and therefore has not previously been addressed in the literature on local polynomial regression. Such uniform convergence results are especially useful when the conditional distribution estimator is the first stage of a semiparametric estimator. We demonstrate the usefulness of these uniform results with two examples: the stochastic equicontinuity condition in $y$, and the estimation of the integrated conditional distribution function.

	\vspace{2em}

	\noindent \textbf{Keywords:} Uniform Bias Expansion, Uniform Convergence Rate, Uniform Asymptotic Linear Representation.
\end{abstract}

\section{Introduction}

This paper studies the nonparametric estimation of the conditional distribution function. The analysis concerns a random variable $Y \in \mathbb{R}$ and a random vector of covariates $X \in \mathbb{R}^d$. The conditional distribution function of $Y$ given $X=x$ is denoted by $F(\cdot|x)$, that is, 
\begin{align*}
	F(y|x) = \mathbb{P}(Y \leq y | X = x), y \in \mathbb{R}.
\end{align*}
When the conditional distribution function $F(\cdot|\cdot)$ is assumed to be smooth, it is natural to consider using the \textit{local linear regression} (LLR) method to estimate $F$.

The main subject of this study is the uniform convergence of the LLR estimator with respect to both $y$ and $x$. In particular, we derive the uniform bias expansion, characterize the uniform convergence rate, and present the uniform asymptotic linear representation of the estimator. As explained in, for example, \cite{hansen2008uniform} and \cite{kong2010uniform}, these uniform results are often useful for semiparametric estimation based on nonparametrically estimated components. 

The estimation of the conditional distribution is an important area of research. \cite{hansen2004nonparametric} studies the asymptotic properties of both the Nadaraya-Watson (local constant) estimator and the LLR estimator, and obtains point-wise convergence results. It is well-known that the LLR estimator has the better boundary properties of the two, but unlike the Nadaraya-Watson estimator, the LLR estimator is not guaranteed to be a proper distribution function.\footnote{To solve this problem \cite{hall1999methods} propose a weighted Nadaraya-Watson estimator that has the same asymptotic distribution as the LLR estimator, but these weights require extensive computation.}
Recently, \cite{politis2020nonparametric} propose a way to correct the LLR estimator. The conditional distribution estimation is also useful for estimating conditional quantiles. For example, \cite{yu1997thesis} and \cite{yu1998local} first estimate the conditional distribution function and then invert it to obtain the conditional quantile function.

The local polynomial estimators have been studied extensively, but the uniform convergence results for the estimation of $F$ are new to the literature. In the general setup of local polynomial estimators, there is only one regressand, namely, $Y$. However, in the conditional distribution estimation, there is a class of regressands, namely, $\mathbf{1}\{Y \leq y\}, y \in \mathbb{R}$. For example, \cite{masry1996multivariate} establishes the uniform convergence rate for general local polynomial estimators, but the uniformity is with respect to the values of the regressors. Therefore, their results can only be applied to an estimate of $F(y|\cdot)$ for a fixed $y \in \mathbb{R}$. For the same reason, the results in \cite{kong2010uniform} cannot be used to provide a uniform asymptotic linear representation for $y \in \mathbb{R}$. Our paper aims to solve these issues and prove that under suitable conditions, the desired results are uniform with respect to both $y$ and $x$. We make use of the recent discovery by \cite{fan2016multivariate} on the support of the covariates, ensuring that the uniform results are valid over the entire support.

The second contribution of the paper is the presentation of a novel way of proving the uniform convergence rate 
via empirical process theory. This theory was developed by \cite{GINE2001consistency} and \cite{GINE2002rates} and supports the uniform almost sure convergence of the kernel density estimator. In this paper, we simplify their method and make it more accessible to users who are only concerned with the notion of uniform convergence in probability.

The remaining parts of the paper are organized as follows. Section \ref{sec:model} introduces the statistical model and the assumptions. Section \ref{sec:bias} establishes the uniform bias expansion result. Section \ref{sec:stochastic} introduces empirical process theory and uses it to prove the uniform convergence rate. Section \ref{sec:asymptotic-linear} presents the uniform asymptotic linear representation and provides a simple example to illustrate the result. The proofs for the theorems in the main text are found in Appendix \ref{sec:proof}. Appendix \ref{sec:ept} contains some preliminary results for empirical process theory.

\section{Model and Assumptions} \label{sec:model}

Let $\{(Y_i,X_i), 1 \leq i \leq n\}$ be a random sample of $(Y,X)$.
The estimation procedure is described as follows. Let $w$ and $k$ be two kernel functions and $K(v) = \int_{-\infty}^v k(u) du$. Let $h_1 = h_{1n} = o(1)$ and $h_2 = h_{2n} = o(1)$ be two scalar sequences of bandwidths. Let $\bm{r}(u) = (1,u^\top)^\top, u \in \mathbb{R}^d$ and $\bm{e}_0 = (1, 0, \cdots, 0)$ be the first $(d+1)$-dimensional unit vector.
The proposed estimator is $\hat{F} (y|x) = \bm{e}_0^\top \bm{\hat{\beta}}(y,x,h_1,h_2)$, where
\begin{align} \label{eqn:def-beta-hat}
	\bm{\hat{\beta}}(y,x,h_1,h_2) & = \left( \hat{\bm{\beta}}_0(y,x,h_1,h_2),\hat{\bm{\beta}}_1(y,x,h_1,h_2), \cdots, \hat{\bm{\beta}}_d(y,x,h_1,h_2) \right)^\top \nonumber\\
	& = \argmin_{\bm{\beta} \in \mathbb{R}^{d+1}} \sum_{i=1}^n \left( K\left( \frac{y - Y_i}{h_2} \right) - \bm{r}(X_i - x)^\top \bm{\beta}  \right)^2 w \left( \frac{X_i - x}{h_1} \right).
\end{align}
Let $H_1$ be the $(d+1) \times (d+1)$ diagonal matrix with diagonal elements: $(1,h_1,\cdots,h_1)$.
The first-order condition of the above minimization problem gives
\begin{align} \label{eqn:Xihat-upsilonhat}
	H_1 \hat{\bm{\beta}}(y,x,h_1,h_2) = \hat{\Xi}(x,h_1)^{-1} \hat{\bm{\upsilon}}(y,x,h_1,h_2),
\end{align}
where
\begin{align*}
	\hat{\Xi}(x,h_1) & = \frac{1}{nh_1^d } \sum_{i=1}^n \bm{r}\left( \frac{X_i - x}{h_1} \right) \bm{r}\left( \frac{X_i - x}{h_1} \right)^\top w \left( \frac{X_i - x}{h_1} \right), \\
	\hat{\bm{\upsilon}}(y,x,h_1,h_2) & = \frac{1}{nh_1^d } \sum_{i=1}^n \bm{r}\left( \frac{X_i - x}{h_1} \right) K\left( \frac{y - Y_i}{h_2} \right) w \left( \frac{X_i - x}{h_1} \right).
\end{align*}

In the construction of the estimator, we do not use the indicator $\mathbf{1}\{ Y_i \leq y \}$. Instead, we use the smoothed version $K\left( (y - Y_i)/h_2 \right)$, which requires the selection of another bandwidth $h_2$ and additional smoothness assumptions on the conditional distribution function. However, there are several advantages to using the smoothed version. First, the estimator constructed from the indicators is not smooth in $y$. When we believe that the true distribution function is smooth, it is customary to use the smoothed estimator. Second, from the asymptotic perspective, the indicator $\mathbf{1}\{ Y_i \leq y \}$ can be considered to be the limiting case of $K\left( (y - Y_i)/h_2 \right)$ for $h_2=0$. As \cite{hansen2004nonparametric} shows, the asymptotic mean squared error is strictly decreasing for $h_2=0$; hence, there are efficiency gains from smoothing. Third, as the simulation results in \cite{yu1997thesis} and \cite{yu1998local} demonstrate, the estimates are not very sensitive to the value of $h_2$. Lastly, as we show in Section \ref{sec:asymptotic-linear}, the smoothed estimator exhibits a stochastic equicontinuity condition in $y$. This condition is particularly useful when the conditional distribution estimation is an intermediate step in a semiparametric estimation procedure. For example, \cite{chen2003estimation} provide results on using the stochastic equicontinuity condition to derive the asymptotic distribution of two-step semiparametric estimators.

The following assumptions are maintained throughout the paper.

\begin{assumptionp}{\textbf{X}} [Distribution of $X$] \label{ass:X} 
	The support of $X$, denoted by $\mathcal{X}$, is convex and compact. The marginal density $f_X$ is bounded away from zero on $\mathcal{X}$. The restriction of $f_X$ to $\mathcal{X}$ is twice continuously differentiable. There exist $\lambda_0, \lambda_1 \in (0,1]$ such that for any $x \in \mathcal{X}$ and all $\epsilon \in (0,\lambda_0]$, there is $x' \in \mathcal{X}$ satisfying $B(x',\lambda_1 \epsilon) \subset B(x,\epsilon) \cap \mathcal{X}$, where $B(x,\epsilon)$ denotes the ball centered at $x$ with radius $\epsilon$. 
\end{assumptionp}
\begin{assumptionp}{\textbf{Y}} [Conditional distribution of $Y|X$] \label{ass:Y}
	The conditional distribution function $F(y|x)$ restricted to $\mathbb{R} \times \mathcal{X}$ is twice continuously differentiable in $y$ and $x$. Moreover, this second-order derivative of $F$ restricted to $\mathbb{R} \times \mathcal{X}$ is uniformly continuous.
\end{assumptionp}

\begin{assumptionp}{\textbf{K}} [Kernel functions] \label{ass:kernel}
	\ \begin{enumerate} [label = (\roman*)]
		\item The kernel function $w$ is a product kernel, that is, $w(u) = w_1(u_1) w_2(u_2) \cdots w_d(u_d).$ Each $w_\ell$ (1) is a symmetric density function with compact support $[-1,1]$; (2) has its second moment normalized to one, that is, $\int u_\ell^2 w_\ell(u_\ell) du_\ell = 1$; (3) is positive in the interior of the support $(-1,1)$; and (4) is of bounded variation.

		\item The kernel function $k$ (1) is a symmetric density function with a compact support and (2) has its second moment normalized to one, that is, $\int v^2 k(v) dv = 1$.
	\end{enumerate}

\end{assumptionp}

A brief discussion of the above assumptions is in order. Assumption \ref{ass:X} is introduced by \cite{fan2016multivariate} as a regularity condition on the support $\mathcal{X}$. It ensures that there are sufficient observations around every estimation location, including the boundary points. Assumption \ref{ass:Y} imposes smoothness conditions on the conditional distribution function $F$. Under this assumption, the Hessian matrix of $F$ is uniformly continuous on the compact support $\textit{supp}(Y,X)$. Assumption \ref{ass:kernel} contains standard conditions on the kernel functions $k$ and $w$. The bounded variation condition is imposed for the application of empirical process theory.

\section{Uniform Bias Expansion} \label{sec:bias}

We denote the true value of the conditional distribution function and its derivative with respect to $x$ as
\begin{align*}
	\bm{\beta}^*(y,x) = \left( \bm{\beta}^*_0(y,x), \bm{\beta}^*_1(y,x), \cdots, \bm{\beta}_d^*(y,x) \right)^\top = \left( F (y|x), \nabla_x F (y|x)^\top \right)^\top,
\end{align*}
where $\nabla_x F (y|x) = \left( \frac{\partial}{\partial x_1} F (y|x), \cdots, \frac{\partial}{\partial x_d} F (y|x) \right)^\top$ is the gradient of $F(y|x)$ with respect to $x$.
A convenient way to analyze the estimator $\hat{\bm{\beta}}(y,x,h_1,h_2)$ is to consider it as an estimator of the pseudo-true value defined by
\begin{align} \label{eqn:def-pseudo-true-beta}
	\bar{\bm{\beta}}(y,x,h_1,h_2) & = \left( \bar{\bm{\beta}}_0(y,x,h_1,h_2),\bar{\bm{\beta}}_1(y,x,h_1,h_2), \cdots, \bar{\bm{\beta}}_d(y,x,h_1,h_2) \right)^\top \nonumber \\
	& = \argmin_{\bm{\beta} \in \mathbb{R}^{d+1}} \mathbb{E} \left[ \left( K\left( \frac{y - Y}{h_2} \right) - \bm{r}(X - x)^\top \bm{\beta}  \right)^2 w \left( \frac{X - x}{h_1} \right) \right].
\end{align}
This pseudo-true value $\bar{\bm{\beta}}$ is deterministic and converges to the true value $\bm{\beta}^*$ as $n \rightarrow \infty$.\footnote{The terminology ``pseudo-true'' is adopted from \cite{fan2016multivariate}.}
We can break the asymptotic analysis of $\hat{\bm{\beta}}(y,x,h_1,h_2) - \bm{\beta}^*(y,x)$ into two parts:
\begin{align*}
	\hat{\bm{\beta}}(y,x,h_1,h_2) - \bm{\beta}^*(y,x) = \underbrace{\hat{\bm{\beta}}(y,x,h_1,h_2) - \bar{\bm{\beta}}(y,x,h_1,h_2)}_{\text{stochastic term}} + \underbrace{\bar{\bm{\beta}}(y,x,h_1,h_2) - \bm{\beta}^*(y,x)}_{\text{bias term}}.
\end{align*}
In this section, we study the bias term, which is the difference between the pseudo-true value and the true value. The first-order condition of (\ref{eqn:def-pseudo-true-beta}) gives an explicit expression of the pseudo-true value: $H_1 \bar{\bm{\beta}}(y,x,h_1,h_2) = \Xi(x,h_1)^{-1} \bm{\upsilon}(y,x,h_1,h_2)$,
where
\begin{align*}
	\Xi(x,h_1) & = \frac{1}{h_1^d}\mathbb{E} \left[ \bm{r}\left( \frac{X - x}{h_1} \right) \bm{r}\left( \frac{X - x}{h_1} \right)^\top w \left( \frac{X - x}{h_1} \right) \right] , \\
	\bm{\upsilon}(y,x,h_1,h_2) & = \frac{1}{h_1^d} \mathbb{E} \left[ \bm{r}\left( \frac{X - x}{h_1} \right) K\left( \frac{y - Y}{h_2} \right) w \left( \frac{X - x}{h_1} \right) \right].
\end{align*}
Define $\Omega(x,h_1) = \int \bm{r}(u) \bm{r}(u)^\top w(u) \mathbf{1}\{ x+ h_1 u \in \mathcal{X}\} du.$ The following lemma shows that the matrices $\Xi(x,h_1)$ and $\Omega(x,h_1)$ are always bounded and invertible.
\begin{lemma} \label{lm:inverse-matrix-well-defined}
	Under Assumptions \ref{ass:X} and \ref{ass:kernel}, there exists $C > 0$ such that the eigenvalues of $\Xi(x,h_1)$ and $\Omega(x,h_1)$ are in $[1/C,C]$ for all $x \in \mathcal{X}$ and $h_1 \geq 0$ small enough. 
\end{lemma}

\begin{theorem} \label{thm:bias}
	Let Assumptions \ref{ass:X}, \ref{ass:Y}, and \ref{ass:kernel} hold. Then
	\begin{align} \label{eqn:bias-general}
		& H_1 \left( \bar{\bm{\beta}}(y,x) - \bm{\beta}^*(y,x) \right) \nonumber \\
		= & \frac{h_1^2}{2}  \Omega(x,h_1)^{-1} \sum_{\ell,\ell' = 1}^d \frac{\partial^2 }{\partial x_\ell \partial x_{\ell'}} F(y|x) \int \bm{r}(u) u_{\ell} u_{\ell'} w(u) \mathbf{1}\{x + h_1 u \in \mathcal{X}\} du \nonumber \\
		+ & \frac{h_2^2}{2} \Omega(x,h_1)^{-1} \frac{\partial^2 }{\partial y^2} F(y|x) \int \bm{r}(u) w(u) \mathbf{1}\{x + h_1 u \in \mathcal{X}\} du + o(h_1^2 + h_2^2),
	\end{align}
	uniformly over $y \in \mathbb{R}$ and $x \in \mathcal{X}$.
	In particular, we have
	\begin{align} \label{eqn:bias-interior}
		\bar{\bm{\beta}}_0(y,x) - \bm{\beta}_0^*(y,x) = \frac{h_1^2}{2} \sum_{\ell =1}^d \frac{\partial^2}{\partial x_\ell^2} F(y|x) + \frac{h_2^2}{2}  \frac{\partial^2}{\partial y^2} F(y|x) + o(h_1^2 + h_2^2),
	\end{align}
	uniformly over $y \in \mathbb{R}$ and $x \in \mathring{\mathcal{X}}_{h_1}$, where $\mathring{\mathcal{X}}_{h_1} = \{ x \in \mathcal{X} : x \pm h_1 = (x_1 \pm h_1, \cdots, x_d \pm h_1) \in \mathcal{X} \}$
	denotes the set of interior points with respect to the bandwidth $h_1$.

\end{theorem}

The novelty of Theorem \ref{thm:bias} is that it provides a uniform bias expansion for the LLR estimator over the entire region $(y,x) \in \mathbb{R} \times \mathcal{X}$. For the boundary points $x \notin \mathring{\mathcal{X}}_{h_1}$, the bias is $O(h_1^2 + h_2^2)$.
For the interior points $x \in \mathring{\mathcal{X}}_{h_1}$, the bias expression (\ref{eqn:bias-interior}) is the same as in \cite{hansen2004nonparametric} and Chapter 6 of \cite{li2007nonparametric}, which contains the curvature of $F(y|x)$. 

\section{Uniform Convergence Rate} \label{sec:stochastic}

In this section, we derive the uniform convergence rate of the stochastic term $\hat{\bm{\beta}}(y,x,h_1,h_2) - \bar{\bm{\beta}}(y,x,h_1,h_2)$. We make use of empirical process theory, which is a powerful tool for studying the uniform convergence of random sequences. Some auxiliary concepts and results are introduced below.

Let $\mathcal{G}$ be a class of uniformly bounded measurable functions defined on some subset of $\mathbb{R}^d$, that is, there exists $M>0$ such that $|g| \leq M$ for all $g \in \mathcal{G}$. We say $\mathcal{G}$ is \textit{Euclidean} with coefficients $(A,v)$, where $A,v>0$, if for every probability measure $P$ and every $\epsilon \in (0,1]$, $N(\mathcal{G},P,\epsilon) \leq A / \epsilon^{v}$,
where $N(\mathcal{G},P,\epsilon)$ is the $\epsilon$-covering of the metric space $(\mathcal{G},L_2(P))$, that is, $N(\mathcal{G},P,\epsilon)$ is defined as the minimal number of open $\norm{\cdot}_{L_2(P)}$-balls of radius $\epsilon$ and centers in $\mathcal{G}$ required to cover $\mathcal{G}$. By definition, if $\mathcal{G}$ is Euclidean with coefficients $(A,v)$, then any subset of $\mathcal{G}$ is also Euclidean with coefficients $(A,v)$.

The above definition of Euclidean classes is introduced by \cite{nolan1987uprocess}. The same concept is also studied in \cite{gine1999laws}, but they refer to what we call ``Euclidean'' as ``VC.'' There is a slight difference that \cite{nolan1987uprocess} use the $L_1$-norm, while \cite{gine1999laws} use the $L_2$-norm. We ignored the envelope in their definition because we only work with uniformly bounded $\mathcal{G}$. The following lemma is useful for deriving the uniform convergence results.

\begin{lemma} \label{lm:gine}

Let $\xi_1,\cdots,\xi_n$ be an iid sample of a random vector $\xi$ in $\mathbb{R}^d$. Let $\mathcal{G}_n$ be a sequence of classes of measurable real-valued functions defined on $\mathbb{R}^d$. Assume that there is a uniformly bounded Euclidean class $\mathcal{G}$ with coefficients $A$ and $v$ such that $\mathcal{G}_n \subset \mathcal{G} $ for all $n$. Let $\sigma^2_{n}$ be a positive sequence such that $\sigma^2_{n} \geq \sup_{g \in \mathcal{G}_n} \mathbb{E}[g(\xi)^2]$. Then
    \begin{align*}
        \Delta_n = \sup_{g \in \mathcal{G}_n}  \left| \sum_{i=1}^n (g(\xi_i) - \mathbb{E}g(\xi_i)) \right| = O_p \left( \sqrt{n \sigma_n^2 |\log \sigma_n|} + |\log \sigma_n| \right) .
    \end{align*}
    In particular, if $n \sigma_n^2 / |\log \sigma_n| \rightarrow \infty$, then $\Delta_n = O_p \left( \sqrt{n \sigma_n^2 |\log \sigma_n|} \right) .$

\end{lemma}

The above lemma is based on the results developed by \cite{GINE2001consistency} and \cite{GINE2002rates}. These two papers focus on proving the almost sure convergence of kernel density estimators based on empirical process theory. We simplify their method and make it available to users who are only interested in convergence in probability. Based on Lemma \ref{lm:gine}, deriving the uniform convergence rate of kernel-based nonparametric estimators boils down to two parts: proving the relevant function classes are Euclidean and computing a uniform bound for the variance.

Controlling the stochastic term does not require the smoothness of the conditional distribution function $F$. We only need the assumptions regarding the support $\mathcal{X}$ and kernel functions. The following theorem establishes the uniform convergence rate for the stochastic term of the LLR estimator. Then, combining this result with Theorem \ref{thm:bias}, we obtain the uniform convergence rate of the LLR estimator as a corollary.

\begin{theorem} \label{thm:stochastic}
	Let Assumptions \ref{ass:X} and \ref{ass:kernel} hold. If the bandwidth satisfies $n h_1^d / |\log h_1| \rightarrow \infty$, then
	\begin{align} \label{eqn:uniform-rate-stochastic}
		\sup_{y \in \mathbb{R}, x \in \mathcal{X}} \left| H_1 \left(\hat{\bm{\beta}}(y,x,h_1,h_2) - \bar{\bm{\beta}}(y,x,h_1,h_2) \right) \right| = O_p \left( \sqrt{\frac{|\log h_1|}{nh_1^d}} \right).
	\end{align}
\end{theorem}

\begin{corollary} \label{cor:uniform-convergence}
	Let Assumptions \ref{ass:X}, \ref{ass:Y}, and \ref{ass:kernel} hold. Then
	\begin{align*}
		\sup_{y \in \mathbb{R}, x \in \mathcal{X}} \left| \hat{F}(y|x) - F(y|x) \right| = O_p \left( h_1^2 + h_2^2 + \sqrt{\frac{|\log h_1|}{nh_1^d}} \right).
	\end{align*}
\end{corollary}

We want to compare the above uniform convergence result with the one in \cite{masry1996multivariate}. In \cite{masry1996multivariate}, the covariates $X$ are supported on the entire space $\mathbb{R}^d$ while the convergence result is only uniform for $x$ in a compact subset of $\mathbb{R}^d$. In our case, the support $\mathcal{X}$ is compact, and the convergence result is uniform over the entire support $(y,x) \in \mathbb{R} \times \mathcal{X}$.

Corollary \ref{cor:uniform-convergence} shows that the uniformity over $y \in \mathbb{R}$ does not have an impact on the convergence rate. This is similar to the fact that we can uniformly estimate the unconditional distribution function under the $n^{-1/2}$-rate. The conditional distribution estimation is a nonparametric problem concerning only the covariates.

So far we have been studying the smoothed estimator. It is also of interest to study the unsmoothed version defined as $\check{F} (y|x) = \bm{e}_0^\top \bm{\check{\beta}}(y,x,h_1)$, where
\begin{align*} 
	\bm{\check{\beta}}(y,x,h_1) & = \argmin_{\bm{\beta} \in \mathbb{R}^{d+1}} \sum_{i=1}^n \left( \mathbf{1}\{Y_i \leq y\} - \bm{r}(X_i - x)^\top \bm{\beta}  \right)^2 w \left( \frac{X_i - x}{h_1} \right).
\end{align*}
The above minimization problem is constructed by replacing the term $K((y-Y_i)/h_2)$ by the term $\mathbf{1}\{Y_i \leq y\}$. For this estimator, we only need to chose one bandwidth $h_1$. Since $\mathbb{E}[\mathbf{1}\{Y_i \leq y\}|X=x]=F(y|x)$, the bias term of this estimator is only $O(h_1^2)$. The bias associated with the smoothing in $y$ no longer exists. The uniform convergence of $\check{F}$ can be established without the differentiability of $F$ with respect to $y$. The trade-off is that the estimator $\check{F} (y|x)$ itself is not smooth in $y$ even if $F$ is. The stochastic term can be analyzed as before. Then we obtain the following uniform convergence rate for the unsmoothed estimator. Notice that we replace Assumption \ref{ass:Y} by a weaker condition which only requires the smoothness of $F$ with respect to $x$.
\begin{theorem} \label{thm:unsmoothed}
	Let Assumptions \ref{ass:X} and \ref{ass:kernel} hold. Assume that $F(y|x)$ restricted to $\mathbb{R} \times \mathcal{X}$ is twice continuously differentiable in $x$, and the second-order derivative $\nabla_x^\top\nabla_x F$ is uniformly continuous on $\mathbb{R} \times \mathcal{X}$. If the bandwidth satisfies $n h_1^d / |\log h_1| \rightarrow \infty$, then
	\begin{align*}
		\sup_{y \in \mathbb{R}, x \in \mathcal{X}} \left| \check{F}(y|x) - F(y|x) \right| = O_p \left( h_1^2 + \sqrt{\frac{|\log h_1|}{nh_1^d}} \right).
	\end{align*}
\end{theorem}

\section{Uniform Asymptotic Linear Representation} \label{sec:asymptotic-linear}

This section derives the uniform asymptotic linear representation of the smoothed LLR estimator. These results are particularly useful in deriving the asymptotic distribution for complicated estimators.

\begin{theorem} \label{thm:asymptotic-linear}
	Let Assumptions \ref{ass:X}, \ref{ass:Y}, and \ref{ass:kernel} hold. If the bandwidth satisfies that $n h_1^d / |\log h_1| \rightarrow \infty$, $nh_1^{d+4}/|\log h_1|$ bounded, and $h_2 = O(h_1)$, then 
	\begin{align*}
		H_1 \left(\hat{\bm{\beta}}(y,x,h_1,h_2) - \bar{\bm{\beta}}(y,x,h_1,h_2) \right) = \Xi(x,h_1)^{-1} \frac{1}{nh_1^d} \sum_{i=1}^n \bm{s}(Y_i,X_i;y,x,h_1,h_2) + O_p \left( \frac{|\log h_1|}{nh_1^d} \right),
	\end{align*}
	uniformly over $y \in \mathbb{R}$ and $x \in \mathcal{X}$, where
	\begin{align*}
		\bm{s}(Y_i,X_i;y,x,h_1,h_2) & = \bm{r}\left( \frac{X_i - x}{h_1} \right) \left( K\left( \frac{y - Y_i}{h_2}  \right) - \tilde{F}(y|X_i) \right) w \left( \frac{X_i - x}{h_1} \right), \\
		\tilde{F}(y|x ) & = \mathbb{E}[ K( (y - Y)/h_2 ) \mid X=x ]. 
	\end{align*}
\end{theorem}

The asymptotic order of the remainder term, $O_p \left( |\log h_1| / nh_1^d \right)$, is the same as in Equation (13) of \cite{kong2010uniform}. Thus, once more, the uniformity over $y \in \mathbb{R}$ does not have an impact on the convergence rate. Combining the results in Theorem \ref{thm:bias} and \ref{thm:asymptotic-linear} and applying the central limit theorem for triangular arrays, we can show that the LLR estimator is asymptotic normal with some asymptotic bias.

We can use this asymptotic linear representation, together with the smoothness of $K$, to derive the following stochastic equicontinuity condition.
\begin{corollary}
	Let the assumptions of Theorem \ref{thm:asymptotic-linear} hold. Let $\delta_n = o(1)$. Then the following stochastic equicontinuity condition hold.
	\begin{align*}
		& \sup_{|y_1 - y_2| \leq \delta_n, x \in \mathcal{X}}\Big| \big(\hat{\bm{\beta}}_0(y_1,x,h_1,h_2) - \bar{\bm{\beta}}_0(y_1,x,h_1,h_2) \big) - \big(\hat{\bm{\beta}}_0(y_2,x,h_1,h_2) - \bar{\bm{\beta}}_0(y_2,x,h_1,h_2) \big) \Big| \\
		=&  O_p \left( \sqrt{\frac{\log h_1}{nh_1^d}} \frac{\delta_n}{h_2} + \frac{|\log h_1|}{nh_1^d} \right).
	\end{align*}
\end{corollary}

We study another simple example to demonstrate how the uniform asymptotic linear representation can be used. Suppose that $d=1$, $\mathcal{X} = [\underline{x},\bar{x}]$, and $Y$ is supported on $[\underline{y},\bar{y}]$. We want to estimate the integrated conditional distribution $\theta = \int_{\underline{y}}^{\bar{y}} \int_{\underline{x}}^{\bar{x}} F(y|x) dxdy$
with the estimator $\hat{\theta} = \int_{\underline{y}}^{\bar{y}} \int_{\underline{x}}^{\bar{x}} \hat{F}(y|x) dxdy$.
The theorem below gives the asymptotic distribution of the estimator $\hat{\theta}$.

\begin{corollary} \label{thm:example}
	Let Assumptions \ref{ass:X}, \ref{ass:Y}, and \ref{ass:kernel} hold. If the bandwidth satisfies that $\sqrt{n} h_1 / |\log h_1| \rightarrow \infty$, $\sqrt{n}h_1^2 \rightarrow 0$, and $h_2 = O(h_1)$, then $\sqrt{n} (\hat{\theta} - \theta) \overset{d}{\rightarrow} N(0,V)$,
	where 
	\begin{align*}
		V = \int \left( \int \left( \mathbf{1}\{s \leq y\} - F(y|t) \right) dy \right)^2 f(s,t) dtds,
	\end{align*}
	and $f(y,x)$ denotes the joint density of $(Y,X)$.
\end{corollary}

\appendix

\numberwithin{equation}{section}
\numberwithin{lemma}{section}

\section{Proofs} \label{sec:proof}

\begin{proof} [Proof of Lemma 1]
	This lemma is almost the same as Lemma 11 in \cite{fan2016multivariate}. The only difference is that in this paper we allow the kernel to diminish at the boundary of the support but the proof of \cite{fan2016multivariate} nonetheless goes through. In fact, following their steps, we can show that the eigenvalues of $\Xi(x,h_1)$ and $\Omega(x,h_1)$ are larger than 
	\begin{align*}
		\inf_{x \in B(0,1)} \min_{b^\top b = 1} b^\top \left( \int \bm{r}(u) \bm{r}(u)^\top w(u) \mathbf{1}\{ u \in B(x,\lambda_1) \} du \right) b,
	\end{align*}
	which is strictly positive since $w>0$ on $[-1,1]^d$. 
\end{proof}

\begin{proof} [Proof of Theorem 1]
	By the standard change of variables and the law of iterated expectations, we can write
	\begin{align*}
		\Xi(x,h_1) & = \int \bm{r}(u) \bm{r}(u)^\top w(u) f_X(x+ h_1 u) du, \\
		\bm{\upsilon}(y,x,h_1,h_2) & = \int \bm{r}(u) w(u) \tilde{F}(y|x + h_1u) f_X(x + h_1 u) du,
	\end{align*}
	where $\tilde{F}(y|x ) = \mathbb{E}[ K( (y - Y)/h_2 ) \mid X=x ]$. 
	Because $f_X$ is continuously differentiable on $\mathcal{X}$, we have $\Xi(x,h_1) = f_X(x) \Omega(x,h_1) + o(1),$
	uniformly over $x \in \mathcal{X}$.
	Applying change of variables and integration by parts to $\tilde{F}(y|x + h_1 u)$, we have
	\begin{align*}
		\tilde{F}(y|x + h_1 u) & = \int K( (y - y')/h_2 ) f(y' | x+ h_1 u) dy' \\
		& = \int K (v) f(y - h_2 v | x+ h_1 u) h_2 dv \\
		& = \int k(v) F(y - h_2 v | x+ h_1 u) dv.
	\end{align*}
	By Assumption \textbf{Y}, $F(y|x)$ restricted to $\mathbb{R} \times \mathcal{X}$ is twice uniformly continuously differentiable. Then for any $y \in \mathbb{R}$, the following expansion holds:
\begin{align*}
	F(y - h_2 v | x+h_1u) & = F(y| x+ h_1 u) - \frac{\partial }{\partial y} F(y|x+ h_1 u) h_2 v + \frac{1}{2} \frac{\partial^2 }{\partial y^2} F(y|x+ h_1 u) h_2^2 v^2 \\
	& + \frac{1}{2} h_2^2 \left( \frac{\partial^2 }{\partial y^2} F(\tilde{y}|x+ h_1 u) v^2 - \frac{\partial^2 }{\partial y^2} F(y|x+ h_1 u) v^2 \right),
\end{align*}
where $\tilde{y}$ is between $y$ and $y - h_2 v$.
Therefore, 
\begin{align*}
	\tilde{F}(y|x + h_1 u) = F(y | x + h_1 u) + \frac{h_2^2}{2} \frac{\partial^2 }{\partial y^2} F(y|x+ h_1 u) \int v^2 k(v) dv + o(h_2^2),
\end{align*}
uniformly over $y \in \mathbb{R}$ and $x+ h_1 u \in \mathcal{X}$. The remainder term is uniformly $o(h_2^2)$ because $\frac{\partial^2 }{\partial y^2} F$ is a continuous function on the compact set $\textit{supp}(Y,X)$.
Next, by the smoothness of $F(y | x)$ with respect to $x$, we have
\begin{align}
	F(y | x + h_1 u) & = F(y | x) + h_1 u^\top\nabla_x F(y | x) + \frac{h_1^2}{2}  u^\top [\nabla_x^\top \nabla_x F(y | \tilde{x})] u \nonumber \\
	& = \bm{r}(u)^\top H_1 \bm{\beta}^* (y,x) + \frac{h_1^2}{2}  u^\top [\nabla_x^\top \nabla_x F(y | x)] u \nonumber \\
	& + \frac{h_1^2}{2} u^\top\left(  [\nabla_x^\top \nabla_x F(y | \tilde{x})] u -  [\nabla_x^\top \nabla_x F(y | x)] u \right) u \nonumber \\
	& = \bm{r}(u)^\top H_1 \bm{\beta}^* (y,x) + \frac{h_1^2}{2}  u^\top[\nabla_x^\top \nabla_x F(y | x)] u + o(h_1^2), \label{eqn:bias-F-x}
\end{align}
uniformly over $y \in \mathbb{R}$ and $x,x+ h_1 u \in \mathcal{X}$. The remainder term is uniformly $o(h_1^2)$ because $\nabla_x^\top \nabla_x F$ is assumed to be uniformly continuous on $\mathbb{R} \times \mathcal{X}$.
Similarly, we have
\begin{align}
	f_X(x + h_1 u) & = f_X(x) + o(1), \label{eqn:bias-f-x} \\
	\frac{\partial^2 }{\partial y^2} F(y|x + h_1 u) & = \frac{\partial^2 }{\partial y^2} F(y|x ) + o(1), \label{eqn:bias-F-y}
\end{align}
uniformly over $y \in \mathbb{R}$ and $x, x+ h_1 u \in \mathcal{X}$. Therefore,
\begin{align*}
	\bm{\upsilon}(y,x,h_1,h_2) & = \Xi(x,h_1) H_1 \bm{\beta}^*(y,x) + \frac{h_1^2}{2} f_X(x) \int \bm{r}(u) w(u) u^\top[\nabla_x^\top \nabla_x F(y | x)] u \mathbf{1}\{x + h_1 u \in \mathcal{X}\}  du \\
	& + \frac{h_2^2}{2} \frac{\partial^2 }{\partial y^2} F(y|x) f_X(x)\int v^2 k(v) dv \int \bm{r}(u) w(u) \mathbf{1}\{x + h_1 u \in \mathcal{X}\} du + o(h_1^2 + h_2^2),
\end{align*}
uniformly over $y \in \mathbb{R}$ and $x \in \mathcal{X}$. Therefore,
\begin{align*}
	& H_1 (\bar{\bm{\beta}}(y,x,h_1,h_2) - \bm{\beta}^*(y,x)) \\
	= & \frac{h_1^2}{2} \Omega(x,h_1)^{-1} \sum_{\ell,\ell' = 1}^d \frac{\partial^2 }{\partial x_\ell \partial x_{\ell'}} F(y|x) \int \bm{r}(u) u_{\ell} u_{\ell'} w(u) \mathbf{1}\{x + h_1 u \in \mathcal{X}\} du \\
	+ & \frac{h_2^2}{2} \Omega(x,h_1)^{-1} \frac{\partial^2 }{\partial y^2} F(y|x) \int v^2 k(v) dv \int \bm{r}(u) w(u) \mathbf{1}\{x + h_1 u \in \mathcal{X}\} du + o(h_1^2 + h_2^2).
\end{align*}
Then the first claim of the theorem follows. 

When $x \in \mathring{\mathcal{X}}_{h_1}$, $x + h_1 u \in \mathcal{X}$ for all $u \in [-1,1]^d$. In that case, $\Omega(x,h_1)$ becomes the identity matrix because $w_\ell$ is symmetric and has variance one. 
Then the second claim of the theorem follows.

\end{proof}

\begin{proof} [Proof of Lemma 2]
	Let $M > 0$ be the uniform bound of $\mathcal{G}$. Notice that each $\mathcal{G}_n$ is a uniformly bounded (by $M$) Euclidean class with the same coefficients $(A,v)$. Denote 
    \begin{align*}
        \Delta_n^o = \sup_{f \in \mathcal{G}_n}  \left| \sum_{i=1}^n \textit{Rad}_i f(X_i) \right|,
    \end{align*}
    where $\textit{Rad}_i, 1 \leq i \leq n,$ is a sequence of iid Rademacher variables.
    By Proposition 2.1 in \cite{GINE2001consistency}, we have for all $n \geq 1$,
    \begin{align*}
        \mathbb{E} \Delta_n^o \leq C \left( v M \log (AM/\sigma_n) + \sqrt{v} \sqrt{n \sigma_n^2 \log (AM/\sigma_n)} \right) = O\left( \sqrt{n \sigma_n^2 |\log \sigma_n|} + |\log \sigma_n| \right).
    \end{align*}
    By the symmetrization result in, for example, Lemma 2.3.1 of \cite{wellner1996}, we know that $\mathbb{E} \Delta_n \leq 2 \mathbb{E} \Delta_n^o = O\left( \sqrt{n \sigma_n^2 |\log \sigma_n|} + |\log \sigma_n| \right).$
    Then the claimed result follows from the Chebyshev inequality.
\end{proof}

\begin{proof} [Proof of Theorem 2]
	
	We proceed with two steps. Recall the expression of $H_1 \hat{\bm{\beta}}$ in Equation (1). In Step 1, we derive the uniform convergence rate of the numerator $\hat{\bm{\upsilon}}(y,x,h_1,h_2)$. In Step 2, we derive the uniform convergence rate of the denominator $\hat{\Xi}(x,h_1)$.

	\bigskip
	\noindent \textbf{Step 1.} 
	To avoid repetition in the proof, we consider a generic element of the vector $\hat{\bm{\upsilon}}(y,x,h_1,h_2)$:
	\begin{align} \label{eqn:v-hat-def}
		\hat{\bm{\upsilon}}_{\pi}(y,x,h_1,h_2) & = \frac{1}{nh_1^d} \sum_{i=1}^n ((X_{i} - x)/h_1)^\pi K\left( \frac{y - Y_i}{h_2} \right) w \left( \frac{X_i - x}{h_1} \right), \nonumber \\
		& = \frac{1}{nh_1^d} \sum_{i=1}^n K\left( \frac{y - Y_i}{h_2} \right) \prod_{\ell = 1}^d ((X_{i \ell} - x_{\ell})/h_1)^{\pi_\ell} w_\ell \left( \frac{X_{i \ell} - x_{\ell} }{h_1} \right),
	\end{align}
	where $\pi = (\pi_1,\cdots,\pi_d), \pi_\ell \in \{0,1\}, \sum \pi_\ell \leq 1.$
	We want to derive the following uniform convergence rate of $\hat{\bm{\upsilon}}_{\ell}(y,x,h_1,h_2)$:
	\begin{align} \label{eqn:numerator-variance-rate}
		\sup_{y \in \mathbb{R}, x \in \mathcal{X}} \left \lvert \hat{\bm{\upsilon}}_{\ell}(y,x,h_1,h_2) - \mathbb{E} \hat{\bm{\upsilon}}_{\ell}(y,x,h_1,h_2) \right \rvert = O_p\left( \sqrt{|\log h_1|/(nh_1^d)} \right) .
	\end{align}
	By defining
	\begin{align*}
		\psi_n(Y,X;y,x) = K\left( \frac{y - Y_i}{h_2} \right) \prod_{\ell = 1}^d ((X_{i \ell} - x_{\ell})/h_1)^{\pi_\ell} w_\ell \left( \frac{X_{i \ell} - x_{\ell} }{h_1} \right)
	\end{align*}
	and $\Psi_n = \{ \psi_n(\cdot,\cdot;y,x): y \in \mathbb{R}, x \in \mathcal{X} \}$, we can write the LHS of (\ref{eqn:numerator-variance-rate}) as
	\begin{align*}
		\sup_{\psi_n \in \Psi_n} \left| \frac{1}{nh_1^d} \sum_{i=1}^n (\psi_n(Y_i,X_i;y,t) - \mathbb{E} \psi_n(Y_i,X_i;y,t)) \right|,
	\end{align*}
	which can be studied with the empirical process theory introduced previously. Notice that $\psi_n$ and $\Psi_n$ depend on $n$ through the bandwidth $h_1$ and $h_2$.

	Consider a larger class $\Psi$ that does not depend on $n$ defined by the following product:
	\begin{align*}
		\Psi = \Psi_Y \Psi_{X_1} \Psi_{X_2} \cdots \Psi_{X_d},
	\end{align*}
	where
	\begin{align*}
		\Psi_Y & = \{(Y,X) \mapsto K\left( (y-Y)/h \right) : y \in \mathbb{R}, h>0 \}, \\
		\Psi_{X_\ell} & = \{ (Y,X) \mapsto \left( (X_\ell - x_\ell)/h \right)^{\pi_\ell}w_\ell\left( (X_\ell - x_\ell)/h \right)  : x \in \mathcal{X}, h>0 \}, \ell = 1,\cdots,d.
	\end{align*}
	For all $n \geq 1$, $\Psi_{n} $ is a subset of the product class $\Psi $. Then we want to show that $\Psi$ is uniformly bounded and Euclidean. If that is true, then we can appeal to Lemma 2.
	
	In view of Lemma \ref{lm:product-of-Euclidean-is-Euclidean}, we only need to show that $\Psi_{Y}$ and $ \Psi_{X_\ell}$ are uniformly bounded and Euclidean. The class $\Psi_{Y}$ is uniformly bounded by $1$. The function $K$ is of bounded variation on $\mathbb{R}$ since it is the integral of the integrable function $k$ (Corollary 3.33 in \cite{folland1999real}). Then by Lemma \ref{lm:BV-VC}, we know that $\Psi_Y$ is Euclidean. 
	The class $\Psi_{X_\ell}$ is uniformly bounded by $\lVert w_\ell \rVert_\infty$. This is because $w_\ell$ is support on $[-1,1]$ and hence the term in front of $w_\ell$, $(X_\ell - x_\ell)/h$, cannot exceed one in magnitude. To show that $\Psi_{X_\ell}$ is Euclidean, notice that the function $u_\ell \mapsto u_\ell^{\pi_\ell} w_\ell(u_\ell)$ is of bounded variation. This is because on the support of $w_\ell$, $[-1,1]$, both $u_\ell \mapsto u_\ell^{\pi_\ell}$ and $w_\ell$ are of bounded variation. Then their product is also of bounded variation (Theorem 6.9, \citet{mathematical-analysis}). Then we know $\Psi_{X_\ell}$ is Euclidean by appealing to Lemma \ref{lm:BV-VC}.

	Next, we want to derive a uniform variance bound for each $\Psi_n$. By the standard change of variables, we know that $\sup_{ \psi_n \in \Psi_n }  \mathbb{E} [\psi_n(Y,X;y,x)^2]$ is bounded by
	\begin{align*}
		\sup_{x \in \mathcal{X} } \mathbb{E} \left[ w\left( \frac{X-x}{h_1} \right)^2  \right] \leq \sup_{x \in \mathcal{X} } h_1^d \int w(u)^2 f_{X}(x+h_1 u) du \leq h_1^d \lVert f_{X} \rVert_\infty \prod_{\ell = 1}^d \norm{w_\ell}_\infty ,
	  \end{align*}
	where we have used the fact that $K \in [0,1]$, and $w_\ell$ is supported on $[-1,1]$ and integrates to $1$. Therefore, we can define $\sigma^2_{\Psi_n} = h_1^d \lVert f_{X} \rVert_\infty \prod_{\ell = 1}^d \norm{w_\ell}_\infty$ as a uniform variance bound for $\Psi_n$. Under the assumption that $n h_1^d / |\log h_1| \rightarrow \infty$, we can apply Lemma 2 to the sequence $\Psi_n$ and obtain that
	\begin{align*}
		\sup_{\psi_n \in \Psi_n} \left| \frac{1}{nh_1^d} \sum_{i=1}^n (\psi_n(Y_i,X_i;y,x) - \mathbb{E} \psi_n(Y_i,X_i;y,x) \right| &= O_p \left( \frac{\sqrt{n \sigma_{\Psi_n}^2 |\log \sigma_{\Psi_n}|}}{nh_1^d}\right)  = O_p \left( \sqrt{ \frac{|\log h_1|}{nh_1^d} }\right),
	\end{align*}
	which is the desired result specified in Equation (\ref{eqn:numerator-variance-rate}).

	\bigskip
\noindent \textbf{Step 2.}
Following the same procedure as in Step 1, we can show that the uniform convergence rate for each element of the matrix $\hat{\Xi}(x,h_1)$ is also $\sqrt{|\log h_1|/(nh_1^d)}$. We omit the details for brevity. Then by Lemma 1, we know that with probability approaching one, the eigenvalues of $\hat{\Xi}(x,h_1)$ is in $[1/C,C]$. In particular, with probability approaching one, the inverse matrix $\hat{\Xi}(x,h_1)^{-1}$ is well-defined, and its induced 2-norm $\left\lVert \hat{\Xi}(x,h_1)^{-1} \right\rVert_2$ is bounded.
    Then applying Lemma 1 once again, we have
    \begin{align*}
        \sup_{x \in \mathcal{X}} \left\lVert \hat{\Xi}(x,h_1)^{-1} - \Xi(x,h_1)^{-1} \right\rVert_2 & = \sup_{x \in \mathcal{X}} \left\lVert \hat{\Xi}(x,h_1)^{-1} (\Xi(x,h_1) - \hat{\Xi}(x,h_1)) \Xi(x,h_1)^{-1} \right\rVert_2 \\
        & \leq \sup_{x \in \mathcal{X}} \left\lVert \hat{\Xi}(x,h_1)^{-1} \right\rVert_2 \left\lVert \Xi(x,h_1) - \hat{\Xi}(x,h_1) \right\rVert_2 \left\lVert \Xi(x,h_1)^{-1} \right\rVert_2  \\
        & = O_p\left( \sqrt{|\log h_1|/(nh_1^d)} \right),
    \end{align*}
    where the second line follows from the submultiplicativity of the induced 2-norm.
    Combing the above result with Step 1, we obtain
   \begin{align*}
        & \sup_{y \in \mathbb{R}, x \in \mathcal{X}} \left\lVert \hat{\Xi}(x,h_1)^{-1} \hat{\bm{\upsilon}}(y,x,h_1,h_2) - \Xi(x,h_1)^{-1} \bm{\upsilon}(y,x,h_1,h_2) \right\rVert_2 \\
        \leq & \sup_{y \in \mathbb{R}, x \in \mathcal{X}} \left\lVert \hat{\Xi}(x,h_1)^{-1} - \Xi(x,h_1)^{-1} \right\rVert_2 \left\lVert \hat{\bm{\upsilon}}(y,x,h_1,h_2) \right\rVert_2 \\
		+& \sup_{y \in \mathbb{R}, x \in \mathcal{X}} \left\lVert  \hat{\bm{\upsilon}}(y,x,h_1,h_2)- \bm{\upsilon}(y,x,h_1,h_2)\right\rVert_2 \left\lVert \Xi(x,h_1)^{-1} \right\rVert_2 \\
        = & O_p\left( \sqrt{|\log h_1|/(nh_1^d)} \right),
    \end{align*}
    where the last line uses the fact that $\hat{v}$ is uniformly bounded. This proves Equation (5).

\end{proof}

\begin{proof}[Proof of Theorem 3]
	For the unsmoothed estimator, we can now define the pseudo-true value by replacing the term $K((y-Y_i)/h_2)$ with the term $\mathbf{1}\{Y_i \leq y\}$ for the minimization problem defined in (3) in the main text. To derive the bias term, we can follow the proof of Theorem 1 and replace $\tilde{F}$ with $F$ in the definition of $\bm{\upsilon}$. Therefore, the bias term in this case can be controlled by using (\ref{eqn:bias-F-x}) and (\ref{eqn:bias-f-x}) without (\ref{eqn:bias-F-y}), which means that we no longer require the differentiability of $F$ with respect to $y$. The bias term associated with $h_2$ (as in Theorem 1) no longer exist. The remaining bias is $O(h_1^2)$. The stochastic term can be dealt with by using the proof of Theorem 2. We replace the class $\Psi_Y$ by the class of indicator functions $\mathbf{1}\{Y_i \leq y\}, y \in \mathbb{R}$, which is also uniformly bounded by $1$ and is Euclidean. The other parts of the proof remain the same.
\end{proof}

\begin{proof} [Proof of Theorem 4]
	Notice that we can write $H_1 \left(\hat{\bm{\beta}}(y,x,h_1,h_2) - \bar{\bm{\beta}}(y,x,h_1,h_2) \right)$ as
	\begin{align*}
		& \hat{\Xi}(x,h_1)^{-1} \left(\hat{\bm{\upsilon}}(y,x,h_1,h_2) - \hat{\Xi}(x,h_1) H_1 \bar{\bm{\beta}}(y,x,h_1,h_2) \right) \\
		= & \hat{\Xi}(x,h_1)^{-1} \frac{1}{nh_1^d } \sum_{i=1}^n \bm{s}(Y_i,X_i;y,x,h_1,h_2) \\
		+ & \hat{\Xi}(x,h_1)^{-1} \frac{1}{nh_1^d} \sum_{i=1}^n \bm{r}\left( \frac{X_i - x}{h_1} \right) \left( \tilde{F}(y|X_i) - \bm{r}\left( \frac{X_i - x}{h_1} \right)^\top H_1 \bar{\bm{\beta}}(y,x,h_1,h_2) \right) w \left( \frac{X_i - x}{h_1} \right) \\
		= & \Xi(x,h_1)^{-1} \frac{1}{nh_1^d } \sum_{i=1}^n\bm{s}(Y_i,X_i;y,x,h_1,h_2) + \textit{err}_1(y,x) + \textit{err}_2(y,x)
	\end{align*}
	where
	\begin{align*}
		\textit{err}_1(y,x) & = \hat{\Xi}(x,h_1)^{-1} \frac{1}{nh_1^d} \sum_{i=1}^n \bm{r}\left( \frac{X_i - x}{h_1} \right) \left( \tilde{F}(y|X_i) - \bm{r}\left( \frac{X_i - x}{h_1} \right)^\top H_1 \bar{\bm{\beta}}(y,x,h_1,h_2) \right) w \left( \frac{X_i - x}{h_1} \right),\\
		\textit{err}_2(y,x) & = \left( \hat{\Xi}(x,h_1)^{-1} - \Xi(x,h_1)^{-1} \right) \frac{1}{nh_1^d } \sum_{i=1}^n \bm{s}(Y_i,X_i;y,x,h_1,h_2).
	\end{align*}
	
	We use the empirical process theory to derive the uniform convergence rates of $\textit{err}_1$ and $\textit{err}_2$ respectively in the following Step 1 and Step 2.

	\bigskip
	\noindent \textbf{Step 1.}
	Define a sequence of function classes $\Phi_n = \{ \phi_n(\cdot,\cdot;y,x) : y \in \mathbb{R}, x \in \mathcal{X} \}$, where
	\begin{align*}
		\phi_n(Y,X;y,x) = \left( \tilde{F}(y|X) - \bm{r}\left( \frac{X - x}{h_1} \right)^\top H_1 \bar{\bm{\beta}}(y,x,h_1,h_2) \right) \prod_{\ell = 1}^d \left( \frac{X_{\ell} - x_{\ell}}{h_1} \right)^{\pi_\ell} w_\ell \left( \frac{X_{\ell} - x_{\ell}}{h_1} \right)
	\end{align*}
	with $\sum \pi_\ell \leq 1$ as before.
	We want to derive the convergence rate of
	\begin{align*}
		\sup_{\phi_n \in \Phi_n} \left| \frac{1}{nh_1^d} \sum_{i=1}^n \phi_n(Y_i,X_i;y,x) \right|.
	\end{align*}
	Notice that $\phi_n$ is already centered, that is, $\mathbb{E}\phi_n(Y,X;y,x) = 0$, by the first-order condition of (2). Define a larger product class $\Phi$ that does not vary with $n$ by $\Phi = \Phi_Y \Psi_{X_1}  \Psi_{X_2} \cdots  \Psi_{X_d}$,
	where $\Psi_{X_\ell}$ is defined in the proof of Theorem 2 and
	\begin{align*}
		\Phi_Y & = \{ (Y,X) \mapsto ( \mathbb{E} [ K((y-Y)/h) |X]  - \bm{r}(X - x)^\top\bm{\beta} \\
		&  \hspace{8em} \times \mathbf{1}\{ |X_\ell - x_\ell| \leq 1, 1 \leq \ell \leq d \} : y \in \mathbb{R},x \in \mathcal{X}, h>0, \norm{\bm{\beta}}_2 \leq C \} .
	\end{align*}
	To understand the expression of $\Phi_{Y}$, recall that by definition $\tilde{F}(y|X) = \mathbb{E} [ K((y-Y)/h_2) \mid X]$. The term $\bar{\bm{\beta}}(y,x,h_1,h_2)$ is replaced by a general $\bm{\beta} \in \mathbb{R}^{d+1}$ with a bounded norm. This can be done as both the numerator and denominator of $\bar{\bm{\beta}}(y,x,h_1,h_2)$ is bounded. The indicator term $\mathbf{1}\{ |X_\ell - x_\ell| \leq 1, 1 \leq \ell \leq d \}$ comes from the support of $w$. This indicator term is needed for deriving the uniform boundedness. 
	
	For each $n$, we have $\Phi_n \in \Phi$. We want to show that $\Phi$ is a uniformly bounded Euclidean class. Since $\Psi_{X_\ell}$ is proven to be uniformly bounded and Euclidean in Theorem 2, we only need to focus on the class $\Psi_Y$.
	First notice that the class
	\begin{align*}
		\left\{ (Y,x) \mapsto \bm{r}(X - x)^\top\bm{\beta} \mathbf{1}\{ |X_\ell - x_\ell| \leq 1, 1 \leq \ell \leq d \} : y \in \mathbb{R},x \in \mathcal{X}, h>0, \norm{\bm{\beta}}_2 \leq C \right\} 
	\end{align*}
	is uniformly bounded and Euclidean in view of Lemma \ref{lm:finite-dim-is-VC}. By Lemma \ref{lm:conditional-expectation-preserves-Euclidean}, we know that the following class is uniformly bounded and Euclidean:
	\begin{align*}
		\Big\{ (Y,X) \mapsto \mathbb{E} [ K((y-Y)/h) \mid X] \mathbf{1}\{ |X_\ell - x_\ell| \leq 1, 1 \leq \ell \leq d \} : y \in \mathbb{R},  h>0 \Big\}
	\end{align*}
	Then by Lemma \ref{lm:sum-of-Euclidean-is-Euclidean}, we know that $\Phi_{Y}$ is uniformly bounded and Euclidean. Hence, $\Phi$ is uniformly bounded and Euclidean.
	
	Then we want to derive a variance bound for each $\Phi_n$. By the usual change of variables, we have for any $y \in \mathbb{R}$ and $x \in \mathcal{X}$,
	\begin{align*}
	  \mathbb{E}[\phi_n(Y,x;y,x)^2] \leq h_1^d \int \left( \tilde{F}_{Y | X}(y | x + h_1 u) - \bm{r}(u) H_1 \bar{\bm{\beta}}(y,x,h_1,h_2) \right)^2  w(u)^2 f_X(x + h_1u) du .
	\end{align*}
	From the uniform bias expansion results in Theorem 1, we have
	\begin{align*}
		\sup_{y \in \mathbb{R},x \in \mathcal{X}} \left| \tilde{F}_{Y | X}(y | x + h_1 u) - \bm{r}(u) H_1 \bar{\bm{\beta}}(y,x,h_1,h_2) \right| = O(h_1^2 + h_2^2) = O(h_1^2).
	\end{align*}
	Therefore, we can construct a uniform variance bound $\sigma^2_{\Phi_n} = O(h_1^{d + 4})$ for the class $\Phi_n$.
	Then by Lemma 2, we can show that
	\begin{align*}
		\sup_{\phi_n \in \Phi_n} \left| \frac{1}{nh_1^d} \sum_{i=1}^n \phi_n(Y_i,X_i;y,x) \right| & = O_p \left( \left(\sqrt{ nh_1^{d+4} |\log h_1| } + |\log h_1| \right)/(n h_1^d) \right) = O_p\left( \frac{|\log h_1|}{nh_1^d} \right),
	\end{align*}
	where the second line follows from the assumption that $n h_1^{d+4} /|\log h_1| \leq C$. Therefore, 
	\begin{align*}
		\sup_{y \in \mathbb{R}, x \in \mathcal{X}} \lVert \textit{err}_1(y,x) \rVert_2 = \sup_{x \in \mathcal{X}} \left\lVert \hat{\Xi}(x,h_1)^{-1} \right\rVert_2  O_p\left( |\log h_1| / (nh_1^d) \right) = O_p\left( |\log h_1| / (nh_1^d) \right) .
	\end{align*}
	
	\bigskip
	\noindent \textbf{Step 2.}
	Similar as before, we can show that
	\begin{align*}
		\sup_{y \in \mathbb{R}, x \in \mathcal{X}} \left\lVert \frac{1}{nh_1^d} \sum_{i=1}^n \bm{s}(Y_i,X_i;y,x,h_1,h_2)\right\rVert_2 = O_p\left( \sqrt{|\log h_1| / (nh_1^d)} \right).
	\end{align*}
	It is straightforward to see that the summand is centered, and the relevant function classes are uniformly bounded and Euclidean. For the variance bound, we can simply bound the term $\left( K\left((y-Y_i)/h_2 \right) - \tilde{F}(y |X_i) \right)^2$ by $1$. We omit the details of the derivation.
	Then by the uniform convergence rate of $\hat{\Xi}(x,h_1)^{-1}$ derived in the proof of Theorem 2, we have
	\begin{align*}
		\sup_{y \in \mathbb{R}, x \in \mathcal{X}} |\textit{err}_2(y,x,h_1,h_2)| = O_p\left( |\log h_1| / (nh_1^d) \right) .
	\end{align*}
	Therefore, we have shown that both the terms $\textit{err}_1(y,x)$ and $\textit{err}_2(y,x)$ are $O_p\left( |\log h_1| / (nh_1^d) \right)$ uniformly. Then the desired result follows.
	
	\end{proof}

	\begin{proof} [Proof of Corollary 2]
		By the asymptotic linear representation in Theorem 4 and the mean value theorem, we have 
		\begin{align*}
			& \sup_{|y_1 - y_2| \leq \delta_n, x \in \mathcal{X}}\Big| \big(\hat{\bm{\beta}}_0(y_1,x,h_1,h_2) - \bar{\bm{\beta}}_0(y_1,x,h_1,h_2) \big) - \big(\hat{\bm{\beta}}_0(y_2,x,h_1,h_2) - \bar{\bm{\beta}}_0(y_2,x,h_1,h_2) \big) \Big| \\
			=  & \sup_{|y_1 - y_2| \leq \delta_n, x \in \mathcal{X}} \Big| \bm{e}_0^\top\Xi(x,h_1)^{-1} \frac{1}{nh_1^d} \sum_{i=1}^n (\bm{s}(Y_i,X_i;y_1,x,h_1,h_2) - \bm{s}(Y_i,X_i;y_2,x,h_1,h_2)) \Big| \\
			& + O_p \left( \frac{|\log h_1|}{nh_1^d} \right) \\
			\leq & C \delta_n \sup_{y \in \mathbb{R}, x \in \mathcal{X}} \Big\lVert\frac{1}{nh_1^d} \sum_{i=1}^n \frac{\partial}{\partial y } \bm{s}(Y_i,X_i;y,x,h_1,h_2) \Big\rVert_2 + O_p \left( \frac{|\log h_1|}{nh_1^d} \right),
		\end{align*}
		where we have used the fact that $\lVert \Xi(x,h_1)^{-1} \rVert_2$ is bounded (Lemma 1). The partial derivative $\frac{\partial}{\partial y } \bm{s}$ is equal to  
		\begin{align*}
			\frac{\partial}{\partial y } \bm{s}(Y_i,X_i;y,x,h_1,h_2) = \frac{1}{h_2} \bm{r}\left( \frac{X_i - x}{h_1} \right) \left( k\left( (y - Y_i)/h_2  \right) - \mathbb{E} \left[ k\left( (y - Y_i)/h_2  \right) \mid X_i \right] \right)  w \left( \frac{X_i - x}{h_1} \right).
		\end{align*}
		Similar as before, we can use Lemma 2 to show that
		\begin{align*}
			& \sup_{y \in \mathbb{R}, x \in \mathcal{X}} \Big\lVert\frac{1}{nh_1^d} \sum_{i=1}^n \bm{r}\left( \frac{X_i - x}{h_1} \right) \left( k\left( (y - Y_i)/h_2  \right) - \mathbb{E} \left[ k\left( (y - Y_i)/h_2  \right) \mid X_i \right] \right)  w \left( \frac{X_i - x}{h_1} \right) \Big\rVert_2 \\
			= & O_p \left( \sqrt{\frac{| \log h_1|}{nh_1^d}} \right).
		\end{align*}
		We omit the details here. It then follows that 
		\begin{align*}
			\sup_{y \in \mathbb{R}, x \in \mathcal{X}} \Big\lVert\frac{1}{nh_1^d} \sum_{i=1}^n \frac{\partial}{\partial y } \bm{s}(Y_i,X_i;y,x,h_1,h_2) \Big\rVert_2 = O_p \left( \sqrt{\frac{| \log h_1|}{nh_1^d}} \frac{1}{h_2} \right).
		\end{align*}
		This proves the corollary.
	\end{proof}

	\begin{proof} [Proof of Corollary 3]
		Consider the following bias-variance decomposition of $\hat{\theta} - \theta$:
		\begin{align*}
			 \underbrace{\int_{\underline{y}}^{\bar{y}} \int_{\underline{x}}^{\bar{x}} \left( \bar{\bm{\beta}}_0(y,x,h_1,h_2) - \bm{\beta}^*_0(y,x,h_1,h_2) \right) dxdy}_{\text{bias term}} + \underbrace{\int_{\underline{y}}^{\bar{y}} \int_{\underline{x}}^{\bar{x}} \left( \hat{\bm{\beta}}_0(y,x,h_1,h_2) - \bar{\bm{\beta}}_0(y,x,h_1,h_2) \right) dxdy}_{\text{stochastic term}}
		\end{align*}
		By Theorem 1 and the assumption that $\sqrt{n}h_1^2 = o(1)$, we know that the bias term is $o(n^{-1/2})$.
		For the stochastic term, we first want to take care of the matrix $\Xi(x,h_1)$. Recall that when $x \in \mathring{\mathcal{X}}_h = [\underline{x}+h, \bar{x}-h]$, $\Xi(x,h_1)$ is equal to the identity matrix $\bm{I}$. In the proof of Theorem 4, we have shown that
		\begin{align*}
			\sup_{y \in \mathbb{R}, x \in \mathcal{X}} \left\lVert \frac{1}{nh_1} \sum_{i=1}^n \bm{s}(Y_i,X_i;y,x,h_1,h_2)\right\rVert_2 = O_p\left( \sqrt{|\log h_1| / (nh_1)} \right).
		\end{align*}
		Therefore, we have
		\begin{align*}
			\norm{ \int_{\underline{y}}^{\bar{y}} \int_{\underline{x}}^{\bar{x}} \left( \Xi(x,h_1) - \bm{I} \right) \frac{1}{nh_1} \sum_{i=1}^n \bm{s}(Y_i,X_i;y,x,h_1,h_2) }_2 = o_p(1/\sqrt{n}).
		\end{align*}
		By Theorem 4 and the assumption that $\sqrt{n}h_1 / |\log h_1| \rightarrow \infty$, we can write the stochastic term as
		\begin{align*}
			\int_{\underline{y}}^{\bar{y}} \int_{\underline{x}}^{\bar{x}} \left( \hat{\bm{\beta}}_0(y,x,h_1,h_2) - \bar{\bm{\beta}}_0(y,x,h_1,h_2) \right) dxdy = \frac{1}{n} \sum_{i=1}^n Z_i  + o_p(1/\sqrt{n})
		\end{align*}
		where 
		\begin{align*}
			Z_i & = \frac{1}{h_1 } \int_{\underline{y}}^{\bar{y}} \int_{\underline{x}}^{\bar{x}} \bm{e}_0^\top \bm{I} \bm{s}(Y_i,X_i;y,x,h_1,h_2) dxdy\\
			& = \frac{1}{h_1 } \int_{\underline{y}}^{\bar{y}} \int_{\underline{x}}^{\bar{x}} \left( K\left( \frac{y - Y_i}{h_2}  \right) - \tilde{F}(y|X_i) \right) w \left( \frac{X_i - x}{h_1} \right) dxdy.
		\end{align*}
		By the standard change of variables, we can write $Z_i$ as
		\begin{align*}
			Z_i = \int_{\underline{y}}^{\bar{y}} \int_{(X_i - \underline{x})/h_1}^{(X_i - \bar{x})/h_1} \left( K\left( \frac{y - Y_i}{h_2}  \right) - \tilde{F}(y|X_i) \right) w \left( u \right) dudy
		\end{align*}
		The random variables $\{ Z_i : 1 \leq i \leq n \}$ forms an iid triangular array. Each $Z_i$ is centered, that is, $\mathbb{E}[Z_i] = 0$. Denote the variance of $Z_i$ as $V_n$, which can be calculated based on change of variables:
		\begin{align*}
			V_n = \mathbb{E}[Z_i^2] & = \int \left( \int_{\underline{y}}^{\bar{y}} \int_{(t - \underline{x})/h_1}^{(t - \bar{x})/h_1} \left( K\left( \frac{y - s}{h_2}  \right) - \tilde{F}(y|t) \right) w \left( u \right) dudy \right)^2 f(s,t) dtds \\
			& = \int \left( \int_{\underline{y}}^{\bar{y}} \int_{-1}^{1} \left( K\left( \frac{y - s}{h_2}  \right) - \tilde{F}(y|t) \right) \mathbf{1}_{[(t - \underline{x})/h_1,(t - \bar{x})/h_1]}(u) w \left( u \right) dudy \right)^2 f(s,t) dtds
		\end{align*}
		As $n \rightarrow \infty$, we have the pointwise convergence results $K( \frac{y - s}{h_2}) \rightarrow \mathbf{1}\{s \leq y\}, \tilde{F}(y|t) \rightarrow F(y|t),$ and $\mathbf{1}_{[(t - \underline{x})/h_1,(t - \bar{x})/h_1]}(u) \rightarrow 1, u \in [-1,1]$. We know that these functions are bounded and the support of $(Y,X)$ is compact. Then by the dominated convergence theorem, we have
		\begin{align*}
			V_n & \rightarrow \int \left( \int_{\underline{y}}^{\bar{y}} \int_{-1}^{1} \left( \mathbf{1}\{s \leq y\} - F(y|t) \right) w \left( u \right) dudy \right)^2 f(s,t) dtds \\
			& = \int \left( \int \left( \mathbf{1}\{s \leq y\} - F(y|t) \right) dy \right)^2 f(s,t) dtds = V,
		\end{align*}
		where the second line follows from the fact that $w$ integrates to 1.
		It is straightforward to see that $Z_i$ is bounded, and hence any moment of $|Z_i|$ is finite. Then we can apply the Lyapnov central limit theorem (for example, Theorem 5.11 in \cite{white2001asymptotic}) to obtain that $\sum Z_i / \sqrt{n}$ converges in distribution to $N(0,V)$.
		The desired result is thus proved.

	\end{proof}

\section{Preliminary Results in Empirical Process Theory} \label{sec:ept}

\begin{lemma} \label{lm:BV-VC}
    Let $K:\mathbb{R} \rightarrow \mathbb{R}$ be a function of bounded variation. Then the following class is Euclidean: 
    \begin{align*}
        \left\{ K\left( (\cdot - x)/h \right): x \in \mathbb{R}, h>0  \right\}.
    \end{align*}
\end{lemma}

\begin{proof} [Proof of Lemma \ref{lm:BV-VC}]
    This is a direct application of Lemma 22(i) in \citet{nolan1987uprocess}.
\end{proof}

\begin{lemma} \label{lm:finite-dim-is-VC}
    Any uniformly bounded and finite-dimensional vector space of functions is Euclidean.
\end{lemma}

\begin{proof} [Proof of Lemma \ref{lm:finite-dim-is-VC}]
    This follows from Lemma 2.6.15 and Theorem 2.6.7 in \citet{wellner1996}.
\end{proof}

\begin{lemma} \label{lm:conditional-expectation-preserves-Euclidean}
    Let $\mathcal{G}$ be a uniformly bounded Euclidean class with coefficients $(A,v)$. Then the class $\{ \mathbb{E}[g(\cdot) \mid X] : g \in \mathcal{G} \}$ is also uniformly bounded and Euclidean with coefficients $(A,v)$.
\end{lemma}
\begin{proof} [Proof of Lemma \ref{lm:conditional-expectation-preserves-Euclidean}]
    This follows from the fact that the conditional expectation is a projection in the Hilbert space $L_2(P)$ and hence reduces the norm.
\end{proof}

\begin{lemma} \label{lm:sum-of-Euclidean-is-Euclidean}
    Let $\mathcal{G}_1$ and $\mathcal{G}_2$ be two classes of functions that are uniformly bounded and Euclidean with coefficients $(A_1,v_1)$ and $(A_2,v_2)$ respectively. Then the class $\mathcal{G}_1 \oplus \mathcal{G}_2 = \{g_1 + g_2:g_1 \in \mathcal{G}_1, g_2 \in \mathcal{G}_2\}$ is also uniformly bounded and Euclidean with coefficients $(A_1A_2A_1 A_2 2^{v_1 + v_2},v_1 + v_2)$.
\end{lemma}

\begin{proof}
    By Inequalities (A.4) in \citet{ANDREWS1994empirical}, we have
    \begin{align*}
        N(\mathcal{G}_1 \oplus \mathcal{G}_2, L_2(P),\epsilon) & \leq N(\mathcal{G}_1 , L_2(P),\epsilon/2) N( \mathcal{G}_2, L_2(P),\epsilon/2) \\
        & \leq A_1 (2/\epsilon)^{v_1} A_2 (2/\epsilon)^{v_2} = A_1 A_2 2^{v_1 + v_2} / \epsilon^{v_1 + v_2}.
    \end{align*}
\end{proof}

\begin{lemma} \label{lm:product-of-Euclidean-is-Euclidean}
    Let $\mathcal{G}_1$ be a class of functions that is uniformly bounded by $M_1$ and Euclidean with coefficients $(A_1,v_1)$ and $\mathcal{G}_2$ a class of functions that is uniformly bounded by $M_2$ and Euclidean with coefficients $(A_2,v_2)$. Then the class $\mathcal{G}_1\mathcal{G}_2 = \{g_1\cdot g_2:g_1 \in \mathcal{G}_1, g_2 \in \mathcal{G}_2\}$ is uniformly bounded by $M_1M_2$ and Euclidean with coefficients $(A_1A_2(M_1 + M_2)^{v_1 + v_2},v_1 + v_2)$.
\end{lemma}

\begin{proof} [Proof of Lemma \ref{lm:product-of-Euclidean-is-Euclidean}]
    The proof is similar to that of Theorem 3 in \citet{ANDREWS1994empirical}. By definition, for every measure $P$ and every $\epsilon \in (0,1]$, $N(\mathcal{G}_1,P,\epsilon) \leq A_1 / \epsilon^{v_1}$ and $N(\mathcal{G}_2,P,\epsilon) \leq A_2 / \epsilon^{v_2}$. We can construct $\{\tilde{g}_{1,j_1}: 1 \leq j_1 \leq J_1\}$ and $\{\tilde{g}_{2,j_2}: 1 \leq j_2 \leq J_2\}$ to be the $\epsilon$-covering of $\mathcal{G}_1$ and $\mathcal{G}_2$, respectively, where $J_1 = N(\mathcal{G}_1,P,\epsilon)$ and $N(\mathcal{G}_1,P,\epsilon)$. For any $g_1 \in \mathcal{G}$ and $g_2 \in \mathcal{G}_2$, suppose $g_1$ is in the $\epsilon$-neighborhood of $\tilde{g}_{1,j_{1,*}}$ and $g_2$ is in the $\epsilon$-neighborhood of $\tilde{g}_{2,j_{2,*}}$. Then the $L_2(P)$ distance between $g_1g_2$ and $\tilde{g}_{1,j_{1,*}}\tilde{g}_{2,j_{2,*}}$ is 
    \begin{align*}
        \left \lVert g_1 g_2 - \tilde{g}_{1,j_{1,*}}\tilde{g}_{2,j_{2,*}} \right \rVert_{L_2(P)} & \leq \left \lVert g_1 g_2 - g_1\tilde{g}_{2,j_{2,*}}\right \rVert_{L_2(P)} + \left \lVert g_1 \tilde{g}_{2,j_{2,*}}-  \tilde{g}_{1,j_{1,*}}\tilde{g}_{2,j_{2,*}} \right \rVert_{L_2(P)} \\
        & \leq M_1  \left \lVert g_2 - \tilde{g}_{2,j_{2,*}}\right \rVert_{L_2(P)} + M_2  \left \lVert g_1 - \tilde{g}_{1,j_{1,*}}\right \rVert_{L_2(P)} \leq (M_1 + M_2) \epsilon.
    \end{align*}
    This means that $\{\tilde{g}_{1,j_1}\tilde{g}_{2,j_2}: 1 \leq j_1 \leq J_1, 1 \leq j_2 \leq J_2\}$ forms a $(M_1 + M_2)\epsilon$-cover of $\mathcal{G}_1\mathcal{G}_2$. Therefore,
    \begin{align*}
        N(\mathcal{G}_1\mathcal{G}_2,L_2(P),\epsilon) & \leq N(\mathcal{G}_1,L_2(P),\epsilon/(M_1 + M_2)) N(\mathcal{G}_2,L_2(P),\epsilon/(M_1 + M_2)) \\
        & \leq A_1 A_2 (M_1 + M_2)^{v_1 + v_2}/ \epsilon^{v_1 + v_2}.
    \end{align*}
    This proves the result.
\end{proof}

\bibliographystyle{apalike}
\bibliography{references.bib}

\end{document}